\documentclass[twoside,leqno,twocolumn,english]{article}
\usepackage{ltexpprt}

\makeatletter
\let\beginalgorithm\@undefined                        
\let\endalgorithm\@undefined                  
\makeatother

\usepackage[ruled, vlined, linesnumbered, nofillcomment]{algorithm2e}

\usepackage{subfigure}
\usepackage{booktabs}
\usepackage{url}
\usepackage{cite}
\usepackage{amsmath}
\usepackage{amssymb}
\usepackage{pgf, pgfnodes}

\newcommand{\set}[1]{\left\{#1\right\}}
\newcommand{\pr}[1]{\left(#1\right)}
\newcommand{\fpr}[1]{\mathopen{}\left(#1\right)}
\newcommand{\spr}[1]{\left[#1\right]}

\newcommand{\abs}[1]{{\left|#1\right|}}

\newcommand{\enset}[2]{\left\{#1 ,\ldots , #2\right\}}

\newcommand{\funcdef}[3]{{#1}:{#2} \to {#3}}

\newcommand{\sfam}[1]{\mathcal{#1}}

\newcommand{\ifam}[1]{\mathcal{#1}}
\newcommand{\iset}[2]{{{#1}\cdots{#2}}}

\newcommand{\prob}[1]{P\fpr{#1}}
\newcommand{\mean}[2]{\operatorname{E}_{#1}\spr{#2}}
\newcommand{\ent}[1]{H\fpr{#1}}

\newcommand{\diag}[1]{\mathrm{diag}\fpr{#1}}

\newcommand{\parents}[1]{\mathit{par}\fpr{#1}}
\newcommand{\opt}[1]{\mathrm{opt}\fpr{#1}}
\newcommand{\seg}[1]{\mathrm{seg}\fpr{#1}}
\newcommand{\model}[1]{M\fpr{#1}}
\newcommand{\score}[1]{s\fpr{#1}}

\newcommand{\pemp}[1]{l\fpr{#1}}
\newcommand{\pratio}[1]{r\fpr{#1}}

\newcommand{\df}[1]{\mathrm{deg}\fpr{#1}}

\newcommand{\define}{\leftarrow}

\newcommand{\pgfnodecircletext}[4]{%
\pgfnodecircle{#1}[stroke]{#2}{#3}%
\pgfnodebox{#1text}[virtual]{\pgfrelative{\pgfxy(0,0)}{\pgfnodecenter{#1}}}{#4}{#3}{#3}}

\newtheorem{example}{Example}

\begin{document}

\title{Are your Items in Order?} 

\author{Nikolaj Tatti}
\date{Department of Mathematics and Computer Science\\Universiteit Antwerpen\\\url{nikolaj.tatti@ua.ac.be}}

\maketitle
\begin{abstract}
Items in many datasets can be arranged to a natural order. Such orders are useful
since they can provide new knowledge about the data and may ease further
data exploration and visualization. 
Our goal in this paper is to define a statistically well-founded and an
objective score measuring the quality of an order. Such a measure can be
used for determining whether the current order has any valuable information
or can it be discarded.

Intuitively, we say that the order is good
if dependent attributes are close to each other. To define the order score we
fit an order-sensitive model to the dataset. Our model resembles a Markov chain
model, that is, the attributes depend only on the immediate neighbors. The
score of the order is the BIC score of the best model. For computing the measure we
introduce a fast dynamic program.
The score is then compared against random orders: if it is better than the
scores of the random orders, we say that the order is good. We also show the
asymptotic connection between the score function and the number of free
parameters of the model. In addition, we introduce a simple greedy approach for
finding an order with a good score. We evaluate the score for synthetic and
real datasets using different spectral orders and the orders obtained with the
greedy method.
\end{abstract}



\section{Introduction}
\label{sec:intro}

Seriation, discovering a linear order for the attributes, is a popular topic in
data mining. The motivation for ordering the attributes comes from the fact
that many datasets have an inherent order, for example, the location of genes
in gene amplification data. In fields such as
paleontology~\cite{fortelius06spectral} or archaeology~\cite{kendall04abudance}
discovering the order (age) of the sites is a fundamental question. There are
many benefits once the order has been discovered. The data can visualized
as a binary matrix for further analysis. Also, the complexity of certain
data mining algorithms, for example discovering
tiles~\cite{gionis04geometric}, can be reduced when we restrict ourselves to
the discovered order.

In this paper, we study measuring the quality of a given attribute order. Such
a measure will help us to determine whether the order at hand is genuinely
significant. We propose an intuitive and novel method for measuring the
quality.  In our approach an order is good if the \emph{attributes that are
dependent of each other are close} in the given order.

\begin{example}
\label{ex:toy}
Assume that we have a dataset with 5 items, $a_1, \ldots, a_{5}$, in which the
value of attribute $a_i$ is generated from the value of the previous attribute
$a_{i - 1}$ by copying it and flipping it with a (small) probability.  We may
now conclude that here the original order is good. Whereas, for example, an
order $a_1, a_4, a_3, a_5, a_2$ is bad since the attributes $a_1$ and $a_2$ are
far away from each other.

\end{example}

Our measure is a generalization of the idea given in the example. Given a
dataset and an order we build a model that depends on the order. Our construction
will be such that model is simple and has good likelihood if the dependent attributes
are close. As a measure for goodness of the model we will use Bayesian Information
Criteria (BIC) which favors simple models that fit data well. It turns out
that we can find the model with the best BIC score through a fast dynamic program.

Also, our model can be seen as a Markov Random Field (MRF) model in which the
the items depend only on their immediate neighbors (see, for example, \cite{cowell99network} for introduction to MRF). 

We compare the score of the discovered model against the scores of random
orders, for example, we consider the probability that a random order will have
a better score than the score of the given model. The probability is close to
$0$, if the order under investigation has an exceptionally good score.

We study the asymptotic behavior of the score and show that 
asymptotically it is an increasing function of the number of free model parameters.
Such a behavior is natural since we favor simple models yet surprising
since the BIC penalty, the term through which the degree of freedom 
affects the score, vanishes as the number of data points grows. 

The rest of the paper is organized as follows. The preliminaries are given in
Section~\ref{sec:preliminaries} and the model itself is defined in
Section~\ref{sec:model}. We explain the dynamic program for finding the best in
Section~\ref{sec:dynamic}. In Section~\ref{sec:order} we consider spectral and
greedy methods for inducing the order. In Section~\ref{sec:permutation} we
compare the score with random orders. We discuss the asymptotic behavior in
Section~\ref{sec:asymptotic}. Section~\ref{sec:related} is devoted to the
related work and in Section~\ref{sec:experiments} we describe our empirical
results. Finally, we conclude the paper with a discussion in
Section~\ref{sec:conclusion}.

\section{Preliminaries and Notation}
\label{sec:preliminaries}
In this section we introduce the preliminaries and notation that we will use in
subsequent sections.

A \emph{transaction} $t \in \set{0, 1}^K$ is a binary vector of length $K$.  A
\emph{binary dataset} is a collection of $N$ transactions having the length
$K$.  We can easily visualize the dataset as a binary matrix of size $N\times
K$.  We use the notation of $\abs{D} = N$ to express the number of transactions
in $D$. An \emph{attribute} $a_i$, $i = 1, \ldots, K$, is a random
Bernoulli variable representing the $i$th element in a random transaction.
We set $A = \enset{a_1}{a_K}$ to be the collection of all attributes.

Assume that we are given a distribution $p$ defined over a space of binary
vectors $\set{0, 1}^K$.  Let $X = \enset{x_1}{x_L} \subseteq A$ be the collection
of attributes. Let $v = \set{0, 1}^L$ be a
binary vector of length $L$. We use the notation $p\fpr{X = v}$ to mean the
probability $p\pr{x_1 = v_1, \ldots, x_L = v_L}$.

Given a binary dataset $D$ we define $q_D$, an \emph{empirical distribution} to
be
\[
q_D\fpr{A = v} = \frac{\abs{\set{t \in D; t = v}}}{\abs{D}}.
\]

We assume that there is a specific linear order induced on the attributes. Such
an order can be identified with a permutation function $o$ mapping from
$\pr{1, \ldots, K}$ to $\pr{1, \ldots, K}$.  To ease the notation we often
assume that $o$ is the identity permutation, that is $o(i) = i$.  Let $X$ be a
collection of attributes, we say that $X$ is an item segment if $X$ contains only consecutive
attributes. For example, $a_{o(1)}a_{o(2)}a_{o(3)}$ is an item segment, however, $a_{o(1)}a_{o(2)}a_{o(4)}$
is not since $a_{o(3)}$ is missing.

The entropy of an item segment $X$ w.r.t. to the distribution $p$, denoted by
$\ent{X; p}$, is
\[
\ent{X; p} = - \sum_v p\fpr{X = v} \log p\fpr{X = v},
\]
where the usual convention $0 \times \log 0 = 0$ is used. All the logarithms in
this paper is of base $2$. We will shorten $\ent{A; p}$ into $\ent{p}$. We also
write $\ent{X; D}$ to mean $\ent{X; q_D}$.

\section{Order-sensitive Model}
\label{sec:model}
In this section we define our model that is based on the order of the
attributes. Informally, our approach is based on generalizing simple markov chain
model demonstrated in Example~\ref{ex:toy}. We generalize this by allowing
the item $a_i$ to depend on several previous items. However, we require that
if $a_i$ depends on $a_j$, then it also must depend on all items between $a_j$
and $a_i$. Thus, our model will be simple if the dependent items are close
to each other.

In order to make the preceeding discussion more formal, 
assume that we are given a cover of item segments $\mathcal{C} =
\enset{C_1}{C_L}$, that is, $\bigcup_i C_i = A$.  We assume that there is no
$C_i, C_j \in \mathcal{C}$ such that $C_i \subset C_j$, that is, $\mathcal{C}$
is an antichain. We also assume that $\mathcal{C}$ is ordered based on the first
attribute of each segment $C_i$. We define $\seg{o}$ to be the family of all
such collections.
Given the collection $\mathcal{C}$ we define a model $\model{\mathcal{C}}$ to be
a collection of distributions that can be expressed as
\begin{equation}
p\fpr{A} = \frac{\prod_{i=1}^L p\fpr{C_i}}{\prod_{i=1}^{L - 1} p\fpr{S_i}} 
= p\fpr{C_1}\prod_{i = 1}^{L - 1} p\fpr{C_{i + 1} - C_i \mid S_i},
\label{eq:probfact}
\end{equation}
where $S_i = C_{i + 1} \cap C_i$.
That is, the attributes in $C_i - C_{i - 1}$ depend only on their immediate
neighbors, $C_i \cap C_{i - 1}$.

\begin{example}
Assume that $\mathcal{C} = \enset{\set{a_1}}{\set{a_K}}$, that is, each segment
is simply a singleton. Then the attributes according to any distribution 
$p \in \model{\mathcal{C}}$ are independent, $p\fpr{A} = \prod_i^K p\fpr{a_i}$.

The distribution $p$ used to generate data in Example~\ref{ex:toy} can be written as
\[
p(a_1)p(a_2 \mid a_1)\cdots p(a_5 \mid a_4)
= \frac{\prod_{i = 1}^4p(a_ia_{i + 1})}{\prod_{i = 2}^4p(a_i)}.
\]
Hence, $p \in \model{\set{a_1a_2, a_2a_3, a_3a_4, a_4a_5}}$.

The other extreme is that $\mathcal{C}$ contains only one segment
containing all items, $\mathcal{C} = \set{A}$. In this case the distribution maximizing
the likelihood is the empirical distribution, $q_D\fpr{A}$.
More generally, if $\mathcal{C}$ is an ordered partition of $A$, that is, $C_i \cap C_{i
- 1} = \emptyset$, then the distribution $p \in \model{\mathcal{C}}$ has
independent components $C_i$, $p\fpr{A} = \prod p\fpr{C_i}$.  Hence we can view
the general model as a generalization of a partition of $A$ by
allowing the segments to overlap.
\end{example}

Given a dataset $D$ we define $p^*$ to be the
the unique distribution from $\model{\mathcal{C}}$ such that $p^*\fpr{C_i = t}
= q_D\fpr{C_i = t}$ for any $C_i \in \mathcal{C}$ and any binary vector $t$ of
length $\abs{C_i}$. We wish to show that $p^*$ maximizes the log-likelihood of $D$.
To see this, let $p \in \model{\mathcal{C}}$. Let us write $S_i = C_i \cap C_{i + 1}$.
Then we have
\[
\begin{split}
&\log p^*(D) - \log p(D) \\
& \quad = \abs{D} \sum_{i = 1}^L \sum_{t} p^*(A = t)\log \frac{p^*(C_i = t_{C_i})}{p(C_i = t_{C_i})}. \\
\end{split}
\]

The last equation is the (scaled) Kullback-Leibler divergence between $p^*$ and
$p$ It is always non-negative and is $0$ if and only if $p = p^*$. Hence $p^*$
maximizes the likelihood.

Let us compute the the log-likelihood of $D$ given a distribution $p^* \in
\model{\mathcal{C}}$, that is, the distribution maximizing the likelihood.  Let
$S_i = C_i \cap C_{i + 1}$. A straightforward calculation reveals that the
log-likelihood can be rewritten as a sum of entropies,

\begin{equation}
\label{eq:likelihood}
\begin{split}
\log p^*\fpr{D} & = \sum_{t \in D} \log p^*\fpr{A = t} \\
& = -\abs{D}\sum_{i=1}^L \ent{C_i; D} + \abs{D}\sum_{i=1}^{L - 1} \ent{S_i; D}.
\end{split}
\end{equation}

Our goal is to find a collection $\mathcal{C}$ for which the model 
provides high likelihood. The problem is that the model with the highest likelihood
would be always for the collection $\mathcal{C} = \set{A}$, for which the
optimal distribution is simply the empirical distribution. To keep it from this behaviour
we will use Bayesian Information Criteria (BIC) to punish more complex models
over the simple ones~\cite{schwarz78bic}. Hence, our goal is to minimize the
BIC cost function, 
\[
-\log p^*(D) + \frac{\log \abs{D}}{2} \df{\mathcal{C}},
\]
where $\df{\mathcal{C}}$ is the number of free parameters in the model
$\model{\mathcal{C}}$.  To compute the number of free parameters, let us
consider the right side of Eq.~\ref{eq:probfact}. The first component $p(C_1)$
can be parameterized with a real vector of length $2^{\abs{C_1}} - 1$.
Similarly, the $i$th component can be parameterized with a vector of length
$2^{\abs{C_i} - \abs{S_{i - 1}}} - 1$.  Since the $i$th component depends on
the values of $S_{i - 1}$ we need parameters for each combination of values of
$S_{i - 1}$. There are $2^{\abs{S_{i - 1}}}$ such values. The total number of free parameters
is then equal to
\begin{equation}
\label{eq:bic}
\begin{split}
\df{\mathcal{C}} & = 2^{\abs{C_1}} - 1 + \sum_{i = 2}^L 2^{\abs{S_{i - 1}}}\pr{2^{\abs{C_i} - \abs{S_{i - 1}}} - 1}   \\
& = \sum_{i=1}^L 2^{\abs{C_i}} - 1 - \sum_{i=1}^{L - 1} 2^{\abs{S_i}} - 1.
\end{split}
\end{equation}

In order to compute the BIC score we can combine Eqs.~\ref{eq:likelihood}--\ref{eq:bic}
in the following way.  We define a score $\score{C}$ for an item segment $C$ to
be
\[
\score{C} = \abs{D}\ent{C; D} + \frac{\log{\abs{D}}}{2} \pr{2^{\abs{C}} - 1}.
\]

Similarly, given a collection $\mathcal{C} = \enset{C_1}{C_L}$ of item segments
we define
\begin{equation}
\label{eq:scoremodeldef}
\score{\mathcal{C}} = \sum_{i=1}^L \score{C_i} - \sum_{i=1}^{L - 1} \score{S_i}, 
\end{equation}
that is, $\score{\mathcal{C}}$ is the sum of the negative likelihood of the optimal
distribution in $\model{\mathcal{C}}$ and the BIC penalty term.
Given an order $o$ we define $\score{o}$ to be the score of the best possible model,
\[
\score{o} = \min \pr{\score{\mathcal{C}} \mid \mathcal{C} \in \seg{o}}.
\]

It is easy to see that minimizing the score produces simple models having high
likelihood of the data. Note that if the dependent attributes are close to each
other, the segments will be short. Hence the BIC punishment will be small.
However, if the dependent attributes are far away, then the segments must be
long and the BIC penalty is far greater. Hence, our score favors orders in
which dependent attributes are close to each other.

\begin{example}
\label{ex:toydata1}
Assume that we have a dataset $D$ given by
\begin{equation}
D = \spr{\begin{array}{ccccc}
1  &  1  &  1  &  0  &  0 \\
1  &  0  &  1  &  0  &  0 \\
0  &  0  &  0  &  1  &  1 \\
1  &  0  &  1  &  1  &  1 \\
0  &  0  &  0  &  0  &  1 \\
\end{array}}.
\label{eq:toydata}
\end{equation}
Assume also that our model is $\mathcal{C} = \set{ab, \, bcd, \, de}$.
We have $S_1 = C_1 \cap C_2 = b$ and $S_2 = C_2 \cap C_3 = d$. The entropies
for the segments are
\[
\begin{split}
\ent{C_1} = 1.52,&\ \ent{C_2} = 2.32, \ent{C_3} = 1.52 , \\
\ent{S_1} = 0.72,& \text{ and } \ent{S_2} = 0.97.
\end{split}
\]
Hence the final score is
\[
\begin{split}
\score{\mathcal{C}} =&  5\times \pr{1.52 + 2.32 + 1.52 - 0.72 - 0.97}  \\
&+ \frac{\log 5}{2}\pr{3 + 7 + 3 - 1 - 1}  = 31.13.
\end{split}
\]
\end{example}

We will finish this section by discussing the connection between the order and
Bayesian networks. Our model can be seen as a Bayesian network, where item
$a_i$ is a parent of $a_j$, if and only if $i < j$ and $a_i$ and $a_j$ belongs
to the same segment (see Figure~\ref{fig:toy} for an example).

\begin{figure}[htb!]
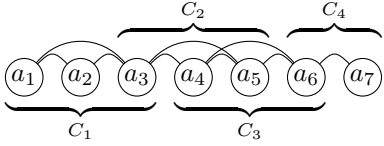

\centering
\begin{pgfpicture}{0cm}{0cm}{4.75cm}{2cm}
\newcommand{\nodesep}{0.75}
\newcommand{\nodewidth}{0.25cm}
\pgfnodecircletext{Node1}{\pgfxy(0,1)}{\nodewidth}{$a_1$}
\pgfnodecircletext{Node2}{\pgfrelative{\pgfxy(\nodesep,0)}{\pgfnodecenter{Node1}}}{\nodewidth}{$a_2$}
\pgfnodecircletext{Node3}{\pgfrelative{\pgfxy(\nodesep,0)}{\pgfnodecenter{Node2}}}{\nodewidth}{$a_3$}
\pgfnodecircletext{Node4}{\pgfrelative{\pgfxy(\nodesep,0)}{\pgfnodecenter{Node3}}}{\nodewidth}{$a_4$}
\pgfnodecircletext{Node5}{\pgfrelative{\pgfxy(\nodesep,0)}{\pgfnodecenter{Node4}}}{\nodewidth}{$a_5$}
\pgfnodecircletext{Node6}{\pgfrelative{\pgfxy(\nodesep,0)}{\pgfnodecenter{Node5}}}{\nodewidth}{$a_6$}
\pgfnodecircletext{Node7}{\pgfrelative{\pgfxy(\nodesep,0)}{\pgfnodecenter{Node6}}}{\nodewidth}{$a_7$}

\pgfsetendarrow{\pgfarrowtriangle{3pt}}
\pgfnodeconncurve{Node1}{Node2}{45}{135}{0.25cm}{0.25cm}

\pgfnodeconncurve{Node2}{Node3}{45}{135}{0.25cm}{0.25cm}
\pgfnodeconncurve{Node1}{Node3}{50}{130}{0.5cm}{0.5cm}

\pgfnodeconncurve{Node3}{Node4}{45}{135}{0.25cm}{0.25cm}

\pgfnodeconncurve{Node4}{Node5}{45}{135}{0.25cm}{0.25cm}
\pgfnodeconncurve{Node3}{Node5}{50}{130}{0.5cm}{0.5cm}

\pgfnodeconncurve{Node5}{Node6}{45}{135}{0.25cm}{0.25cm}
\pgfnodeconncurve{Node4}{Node6}{50}{130}{0.5cm}{0.5cm}

\pgfnodeconncurve{Node6}{Node7}{45}{135}{0.25cm}{0.25cm}

\pgfputat{\pgfrelative{\pgfxy(-0.25,-0.3)}{\pgfnodecenter{Node1}}}{\pgfbox[left,base]{$\underbrace{\hspace{2cm}}_{C_1}$}}
\pgfputat{\pgfrelative{\pgfxy(-0.25,0.5)}{\pgfnodecenter{Node3}}}{\pgfbox[left,base]{$\overbrace{\hspace{2cm}}^{C_2}$}}
\pgfputat{\pgfrelative{\pgfxy(-0.25,-0.3)}{\pgfnodecenter{Node4}}}{\pgfbox[left,base]{$\underbrace{\hspace{2cm}}_{C_3}$}}
\pgfputat{\pgfrelative{\pgfxy(-0.25,0.5)}{\pgfnodecenter{Node6}}}{\pgfbox[left,base]{$\overbrace{\hspace{1.25cm}}^{C_4}$}}

\end{pgfpicture}
\caption{A toy example of an order-sensitive model. The model can be viewed as a special type of Bayes Network.}
\label{fig:toy}
\end{figure}

If we interpret our model as a Bayesian network, then it is easy to see that
$p^*(a_i \mid \parents{a_i}) = q_D(a_i \mid \parents{a_i})$.  In other words,
the distribution $p^*$ is the maximum-likelihood estimate of the corresponding
Bayesian network.  In fact, the factorization of the BIC score in
Eq.~\ref{eq:scoremodeldef} is equivalent to the factorization obtained by
performing a junction tree decomposition (see~\cite{cowell99network}) to the
network. However, interpreting the model as a Bayesian network is somewhat
misleading, since our model does not change if we reverse the order of the
items.

\section{Finding the Optimal Model}
\label{sec:dynamic}
In this section we demonstrate how we can find the collection of item segments
having the minimal score $\score{\mathcal{C}}$. We do this by constructing a
dynamic program. Furthermore, we demonstrate how we can prune a large amount of
long segments, and thus reducing the required execution time. In addition, we
discuss how to efficiently compute the entropy using a particular tree structure.

\subsection{Forming the Dynamic Program}
Our goal is to find a collection $\mathcal{C}$ for which the score
$\score{\mathcal{C}}$ is the smallest possible. We achieve this by using
dynamic program and solving subproblems.
In order to do this we set $f\fpr{i, j}$ to be the best collection of segments
covering the attributes $\enset{a_i}{a_K}$ such that the first segment has the
attribute $a_j$ The optimal collection would be then $f\fpr{1, 1}$.

To compute the values of $f$ we use the following order: 
\[
\begin{split}
f(K, K), f(K - 1, K), f(K - 1, K - 1), \ldots\qquad\\
\qquad\ldots, f(1, K), \ldots, f(1, 1).
\end{split}
\]
To compute $f\fpr{i, j}$ we first note that either $f\fpr{i, j}$ has $X =
\iset{a_i}{a_j}$ as its first segment or $f\fpr{i, j} = f\fpr{i, j + 1}$.  Let
$\mathcal{H}$ be the optimal collection having $X$ as the first segment.  The
second segment of $\mathcal{H}$ must cover $j + 1$ and must start at $k = i +
1, \ldots, j + 1$.  Hence, we have
\[
	\mathcal{H} = \underset{i < k \leq j + 1}{\arg \min}\ \score{X \cup f\fpr{k, j + 1}}.
\]
Once we have discovered $\mathcal{H}$ we can set 
\[
f\fpr{i, j} = \arg \min \pr{\score{\mathcal{H}}, \score{f\fpr{i, j + 1}}}.
\]
The details of the algorithm are given in Algorithm~\ref{alg:dynamic}.

\begin{algorithm}[htb!]
	\For{$i = K, \ldots, 1$} {
		$f\fpr{i, K} \define \iset{a_i}{a_K}$\;
	}
	\For{$i = K - 1, \ldots, 1$} {
		\For{$j = K - 1, \ldots, i$} {
			$X \define \iset{a_i}{a_j}$\;
			$\mathcal{H} \define X \cup f\fpr{i + 1, j + 1}$\;
			\For{$k = i + 2, \ldots, j + 1$} {
				\If{$\score{\mathcal{H}} < \score{X \cup f\fpr{k, j + 1}}$} {
					$\mathcal{H} \define X \cup f\fpr{k, j + 1}$\;
				}
			}
			$f\fpr{i, j} \define \arg \min \pr{\score{\mathcal{H}}, \score{f\fpr{i, j + 1}}}$\;
		}
	}
	\Return $f\fpr{1, 1}$\;

\caption{Dynamic program solving the optimal collection of segments.}
\label{alg:dynamic}
\end{algorithm}

\subsection{Pruning Long Segments}
The running time for Algorithm~\ref{alg:dynamic} is $O\pr{K^3}$. In this
section we introduce a pruning condition and reduce the running time to $O\pr{K\min\pr{K, \log \abs{D}}^2}$.
This improvement is crucial since $\log \abs{D}$ is typically much smaller than
$K$.

We begin by asserting the necessary criteria for a segment occurring in the
optimal collection.
\begin{lemma}
Let $X$, $Y$ be segments of length $N$ having $N - 1$ mutual items. If
\begin{equation}
\label{eq:prune}
\score{X \cup Y} \geq \score{X} + \score{Y} - \score{X \cap Y},
\end{equation}
then there is a collection without $X \cup Y$ having an optimal score.
\end{lemma}

\begin{proof}
Let $\mathcal{C}$ be a collection with the optimal score. Assume that $X \cup Y
\in \mathcal{C}$.  Assume that $X$ and $Y$ are not included in some $S_i$, that
is, the only set in $\mathcal{C}$ that contains $X$ or $Y$ is $X \cup Y$.  We
can build an alternative collection by replacing $X \cup Y$ with separate $X$
and $Y$.  The impact on the score is that we replace the term $\score{X \cup
Y}$ with the terms $\score{X} + \score{Y} - \score{X \cap Y}$.  The assumption
now implies that this alternative model will have a better or an equal score.
Assume now that $X$ is one of the $S_i$ but $Y$ is not. Then if we replace $X
\cup Y$ with $Y$ we replace the terms $\score{X \cup Y} - \score{X}$ with the
terms $\score{Y} - \score{X \cap Y}$.  The assumption now implies that the
new has a better or an equal score.  The case is similar for $Y$.

If $X$ and $Y$ are both included in some $S_i$, then by simply removing $X \cup
Y$ we replace the terms $\score{X \cup Y} - \score{X} - \score{Y}$ with
$-\score{X \cap Y}$. This completes the proof.
\end{proof}

We can use the lemma to prune a large number of segments from dynamic program.
In fact, if the segment is long enough, then the lemma is automatically guaranteed.

\begin{proposition}
Let $X$, $Y$ be segments of length $N$ having $N - 1$ mutual items. If
\[
N \geq \log \abs{D} - \log \log \abs{D} + 2,
\]
then $\score{X \cup Y} \geq \score{X} + \score{Y} - \score{X \cap Y}$.
\label{prop:bound}
\end{proposition}

The idea behind the proof is that the BIC penalty for long segments is
too large when compared to the gain from obtained from the likelihood.

\begin{proof}
Let us write $V = X \cap Y$ and $W = X \cup Y$ and $Z = X - Y$.
By using the definition of the score function we can rewrite the inequality in
Eq.~\ref{eq:prune} as
\begin{equation}
\label{eq:bound}
\begin{split}
&2^{N - 2}\log \abs{D}  \\
&\qquad\geq \abs{D}\pr{\ent{X} + \ent{Y} - \ent{V} - \ent{W}}.
\end{split}
\end{equation}
To guarantee this inequality we will bound the right side from above.
Let $A$ and $B$ be two sets of items.  Basic
properties of the entropy state that $\ent{A} + \ent{B} \geq \ent{A \cup B} \geq \ent{A}$.
This immediately implies that $\ent{Y} - \ent{W} \leq 0$ and that
\[
\begin{split}
\ent{X} - \ent{V} & = \ent{V \cup Z} - \ent{V}  \\
& \leq \ent{V} + \ent{Z} - \ent{V} = \ent{Z} \leq 1.
\end{split}
\]
The last inequality is true since, by definition, $Z = X - Y$ contains only one
item and the entropy of a single Bernoulli variable is $1$, at maximum.  This
implies that the right side of Eq.~\ref{eq:bound} is bounded by $\abs{D}$.
Hence we have the sufficient condition
\[
2^{N - 2}\log \abs{D} \geq \abs{D}.
\]
By taking the logarithm we obtain the assessment of the proposition.
\end{proof}

The proposition tells us that we can safely ignore any segments of length $\log
\abs{D} - \log \log \abs{D} + 3 \in O\pr{\log \abs{D}}$, or longer. Thus, by
modifying Line~2 in Algorithm~\ref{alg:dynamic} we can ignore computing $f(i,
j)$ if $j - i$ is large enough.  This speeds up the execution time of
Algorithm~\ref{alg:dynamic} to $O\pr{K\min \pr{K, \log \abs{D}}^2}$ which can
be very effective for datasets with many attributes but small number of
transactions. 

\subsection{Computing Entropy Efficiently}

In our experiments the bottleneck is the entropy calculation. Consequently, it
is to optimize the computation to be as fast as possible. In this section we
will show that in our case, computing entropy for a single segment can be in
essentially $O(\abs{D})$ time.

Assume that we want to compute entropy $\ent{C ; D}$ for a given item segment
$C$.  To compute this we first partition the transaction into groups
$\enset{T_1}{T_L}$, such that transactions $t$ and $u$ belong to the same group $T_i$
if and only if $t_C = u_C$. Then it follows directly from the definition that
\[
	\ent{C ; D} = -\sum_{i = 1}^L \frac{\abs{T_i}}{\abs{D}} \log \frac{\abs{T_i}}{\abs{D}}.
\]

Constructing the partition for a single item segment from a scratch can be done in
$O(\abs{C}\abs{D})$ time by a radix sort. We can, however, speed up the total
execution time by computing the entropies of several segments simultaneously.
More precisely, if we are given indices $s$ and $e$ such that $s < e$, then
Algorithm~\ref{alg:entropy} will output entropies for segments $a_s\cdots a_j$, where $s \leq j \leq e$.

\begin{algorithm}[htb!]
	$T_1 \define D$\;
	\For{$j = s, \ldots, e$} {
		\ForEach{$T_i$ in the partition} {
			$U \define \emptyset$\;
			\ForEach{$t \in T_i$, $t_j = 1$} {
				Remove $t$ from $T_i$ and add to $U$\;
			}
			\lIf{$T_i = \emptyset$} {Remove $T_i$ from the partition}
			\lIf{$U \neq \emptyset$} {Add $U$ to the partition}
		}
		$\ent{a_s\cdots a_j} \define -\sum_{i = 1}^L \frac{\abs{T_i}}{\abs{D}} \log \frac{\abs{T_i}}{\abs{D}}$\;
	}
\caption{Algorithm for computing entropies $\ent{a_s}$, $\ent{a_sa_{s + 1}}$, \ldots,  $\ent{a_s\cdots a_e}$.}
\label{alg:entropy}
\end{algorithm}

The execution time of Algorithm~\ref{alg:entropy} is $O((e - s + 1)\abs{D})$
but it will compute entropies for $e - s + 1$ segments simultaneously. We use
this algorithm to cache all the needed entropies before we invoke the dynamic
program in Algorithm~\ref{alg:dynamic}. We need only $K$ calls of
Algorithm~\ref{alg:entropy}, one call for each $s = 1, \ldots, K$.

\section{Inducing the Order}
\label{sec:order}
So far we have assumed that we are given an order and we have focused on
measuring the quality of that order.  In this section, we will consider
different techniques for inducing the order from the data.

\subsection{Fiedler Vector Approach}
\label{sec:fiedler}

Assume for the moment that we are interested in a model that have segments only
of size 2. We are interested in finding the order that produces the best
model. We define $C$, the \emph{mutual information matrix} of size
$K \times K$ to be
\[
C_{ij} = \ent{a_i} + \ent{a_j} - \ent{a_ia_j}.
\]

Discovering the best order reduces to Traveling Salesman Problem which is a
computationally infeasible problem~\cite{papadimitriou94complexity}.

We will use a popular technique in which the order is constructed from a Fiedler
vector~\cite{fiedler75fiedler}. To be more precise, let the \emph{Laplacian} of
$C$ be $L = \diag{C} - C$, where $\diag{C}$ is a diagonal matrix containing the
sums of the rows of $C$.  The Fiedler vector $f(C)$ is the eigenvector of $L$
of the second smallest eigenvalue.  The order induced by this vector is simply
the order of indices of the sorted entries of the vector. If we were to
permute the attributes, then entries in the Fiedler vector are shuffled with
the exactly same permutation. Thus, the fiedler 
order is not affected by the original order of the attributes.
We will justify our choice by showing that the Fiedler vector does return the
best order in some cases. To be more specific, assume for a moment that
the attribute $a_i$ depends only on its immediate neighbor $a_{i - 1}$. In
other words, we assume that the data is generated from a model constructed from
the segments $\set{a_{i - 1}a_i \mid i = 2, \ldots, K}$. If that is the case,
then the mutual information matrix $C$ has a special property: the entries of $C$
are decreasing as we move from the diagonal towards the corners. That is, 
\begin{equation}
\max \pr {C_{i(j + 1)}, C_{(i - 1)j}} < C_{ij}, \text{ for } i < j,
\label{eq:monotone}
\end{equation}
and similarly for the lower triangular part of $C$. Such a matrix is called
$R$-matrix~\cite{atkins99seriation}. The following theorem states that for
$R$-matrices, the Fiedler vector finds the correct order.

\begin{theorem}[Theorem 3.3 in~\cite{atkins99seriation}]
\label{thr:monotone}
Let $C$ be such that the property in Eq.~\ref{eq:monotone} holds. Then the Fiedler vector $f(C)$
will have $f_i > f_j$ whenever $i < j$.
\end{theorem}

Motivated by this result we consider in this paper 4 different approaches for
computing the order.
\begin{enumerate}
\item $\textsc{mi} = f(C)$ uses the order obtained from the Fiedler vector of the
mutual information matrix.
\item $\textsc{m2} = f(C')$, where $C'_{ij} = C_{ij}$ except when $C_{ij} \leq
\log\abs{D} / 2\abs{D}$ in which case $C'_{ij} = 0$. The motivation behind this
approach is as follows. The mutual information $C_{ij}$ can be viewed as a
difference of the log-likelihoods.  The first model is the independence model
$\mathcal{M}_1$ and has the log-likelihood $-\ent{a_i} - \ent{a_j}$.  The
second model $\mathcal{M}_2$ is the full contingency table model and has the
log-likelihood $-\ent{a_ia_j}$.  Here the idea is that instead of always
comparing $\mathcal{M}_2$ against $\mathcal{M}_1$, we first select the one
model that is more probable. If we select $\mathcal{M}_2$, then the difference
is the mutual information. If, on the other hand, we select $\mathcal{M}_1$,
then the difference will be $0$. If we use BIC score as a criteria for
selecting the model, then $\mathcal{M}_1$ will have a better score if and only
if $C_{ij} \leq \log\abs{D} / 2\abs{D}$. In other words, if $C_{ij}$ is too small compared
to the BIC penalty, then we treat $a_i$ and $a_j$ as independent, and set the mutual information to be $0$.

\item $\textsc{co} = f(D^TD)$, that is, the Fiedler vector of the co-occurrence
matrix.  Such orders have been used for minimizing the Lazarus effects, that
is, $0$s occurring between $1$s~\cite{atkins99seriation}.
\item $\textsc{cs} = f(VD^TDV)$, where $V$ is a diagonal matrix, such that,
$V_{ii} = (D^TD)_{ii}^{-1/2}$,  that is, \textsc{cs} is the order obtained from
the cosine similarity matrix.
\end{enumerate}

For calculating the Fiedler order we use the algorithm given
in~\cite{atkins99seriation}. We should point out that the Fiedler order is
unique if there is only one Fiedler vector (up to normalization). However, it
is often the case that there are several vectors and hence several orders. In
that case the Atkins' algorithm returns a set of all possible orders
represented by a PQ-tree. This set of orders may be large (it can contain all
possible orders) so in practice we will sample orders from this set in our
experiments. Luckily, sampling orders from a set represented by a PQ-tree is
trivial.

\subsection{Greedy Local Search}
\label{sec:greedy}
In addition to the aforementioned spectral methods we will consider a simple
greedy descent approach. Assume that we are given an order $o$. For each $i =
2, \ldots, K$, we consider orders which are obtained from $o$ by
swapping $o(i)$ and $o(i - 1)$. Among such orders we select the one
that has the lowest score, say $b$. If the score $\score{b}$ is lower than the
original $\score{o}$, then we replace $o$ with $b$ and repeat the step,
otherwise we stop the search and output $o$. The pseudo-code is given in Algorithm~\ref{alg:greedy}.

\begin{algorithm}[htb!]
\While {changes} {
	$b \define o$\;
	\ForEach{$i = 2, \ldots, K$} {
		$u \define o$\;
		Swap $u(i - 1)$ and $u(i)$\;
		\If{$\score{u} < \score{b}$} {
			$b \define u$\;
		}
	}
	$o \define b$\;
}
\Return $o$\;
\caption{\textsc{GreedyOrder}, A simple hill-climbining algorithm for improving the order $o$.}
\label{alg:greedy}
\end{algorithm}

\section{Comparing to Random Orders}
\label{sec:permutation}
The score of a model alone is not sufficient alone and it needs to be compared
against some baseline. In this section we will consider two different
approaches for the post-normalization.

Let $o$ be the order of the attributes. Let $\sfam{I}$ be the collection of
item segments corresponding to the independence model, $\sfam{I} =
\enset{\set{a_1}}{\set{a_K}}$.
Our first attempt is to compare the scores $\score{o}$ and $\score{\sfam{I}}$,
that is, how good the score is against the independence model. This approach,
however, has a drawback. Consider a dataset with two clusters such that the
probability of having $1$ is high in the first cluster and low in the second
cluster. Assume also that the inside a cluster the attributes are independent
and have the same probability of having $1$.  Such a dataset has a curious
property. The best score for any order is lower than the score for the
independence model.
However, since data is symmetric,
the best scores for all orders should be equal.

The discussion above suggests a more refined approach, that is, we should
compare the score of the given order against the scores of random orders.

Instead of defining a measure just for a single order we define a measure
for a \emph{set} of orders. The reason for this is that Atkins' algorithm
(see Section~\ref{sec:fiedler}) may not return a single order but a
set of orders represented as a PQ-tree.

Let us assume that we are given a set of orders $O$. Let $o$ be a random order
from $O$ selected uniformly. Also let $u$ be a random order from the set of
all possible orders $U$.  We define the measure to be the the probability that
$\score{o}$ is higher than $\score{u}$, that is, 
\[
\pemp{O} = \prob{\score{o} > \score{u}} + \prob{\score{o} = \score{u}} / 2.
\]

If the set $O$ consists of orders with exceptionally low scores then the $\pemp{o}$
will be close to $0$.  If the orders have the same score as random orders, then
$\pemp{O}$ will be close to $1/2$. If we are given a single order $o$, we write
$\pemp{o}$ for $\pemp{\set{o}}$.

In practice, we estimate the measure by sampling orders $o$ and $u$.  A problem
with the measure $\pemp{o}$ is that in the case when it is close to $0$, the
number of samples should be really large in order to get an estimate different
from $0$. Hence, we define a second, smoother, measure based on estimation with
normal distributions.

In order to do that, let $\mu_1 = \mean{}{\score{o}}$ and $\mu_2 =
\mean{}{\score{u}}$ be the means of the scores and let $\sigma_1^2 =
\mean{}{\score{o}^2} - \mu_1^2$ and $\sigma_2^2 = \mean{}{\score{u}^2} -
\mu_2^2$ be the variances. Let $X_i$ be a random normal variable distributed as
$N(\mu_i, \sigma_i)$, for $i = 1, 2$. We define
\[
\pratio{O} = -\log \prob{X_1 > X_2} = -\log \Phi\fpr{\frac{\mu_1 - \mu_2}{\sqrt{\sigma_1^2 + \sigma_2^2}}},
\]
where $\Phi$ is the cumulative density function for the standard normal
distribution $N(0, 1)$. In practice we estimate the means and the variances by
sampling the orders from the sets $O$ and $U$. We should make clear that
$\pratio{O}$ should not be treated as an estimate for $-\log \pemp{O}$. The
distribution of the scores for random scores is not normal even if the number
of data points increases to infinity. Nevertheless, the measure $\pratio{O}$
has desired properties. It is large if the scores of orders in $O$ are
significantly better.  If the scores are as good as random, then
$\pratio{O}$ is close to $-\log 1/2 = 1$.

The value $\pemp{o}$ gives us means to test whether the order $o$ is
significantly better than the random order. However, if the order is induced
from the dataset, for example, using the spectral methods,
we may overfit the data. To illustrate the problem
consider the cluster data discussed above. In this data no order should be
significant.  However, since dataset is finite there are small variations in
the scores.  Now consider the algorithm that finds the order with the best
score. The measure $\pemp{o}$ of such order will always be $0$. We remedy this
problem with cross validation, that is, we split the data into two random
datasets. We will learn the order from the first dataset and compute the measure
from the second.

\section{Asymptotic Analysis}
\label{sec:asymptotic}
Assume that we know the distribution $p$ from which data is generated. What is then
the appropriate definition for a score of a linear order? In this section we will
define such a score, which we denote $\df{\opt{o, p}}$, and show that $\score{o}$ is directly connected to this
score as the number of transactions goes to infinity.

Let $p$ be a distribution from which data is generated. Assume that we are given an order
$o$ and let 
\[
\opt{o, p} = \arg \min_{\sfam{C}} \set{\df{\sfam{C}}; \sfam{C} \in \seg{o}, p \in \model{\sfam{C}}},
\]
that is, $\opt{o, p} \in \seg{o}$ is the set of item segments with the smallest
number of freedom such that $p \in \model{\opt{o, p}}$. Such a set exists, since
$p \in \model{\set{A}}$.

\begin{example}
We continue Example~\ref{ex:toy} given in Section~\ref{sec:intro}. The optimal
set of segments for the order $a_1\cdots a_5$ is $\sfam{O}_1 = \set{a_1a_2,
a_2a_3, a_3a_4, a_4a_5}$. On the other hand, the optimal set of segments for the
order $a_1a_4a_3a_5a_2$ is $\sfam{O}_2 = \set{A}$ since there is no other way to
include $a_1$ and $a_2$ into the same segment. The degrees are $\df{\sfam{O}_1} = 9$
and $\df{\sfam{O}_2} = 31$. 
\end{example}

The preceding example demonstrates that if the degree of $\opt{o, p}$ is low,
then the order is good. This suggests that the orders in which the dependent
attributes are close will have a low degree of freedom, hence there should be a
connection between $\df{\opt{o, p}}$ and $\score{o}$.  We will now state the
main result of this section that states essentially that asymptotically
$\score{o}$ is an increasing function of $\df{\opt{o, p}}$. The proof of the
theorem is given in Appendix.

\begin{theorem}
\label{thr:asymptotic}
Let $p$ be a distribution from which $D$ is generated such that $p(A = t) > 0$ to any $t$. Let $o_1$ and $o_2$ be
two orders such that $\df{\opt{o_1, p}} < \df{\opt{o_2, p}}$. Then the
probability of $\score{o_1} < \score{o_2}$ converges to $1$ as $\abs{D}$
approaches infinity.
\end{theorem}

Note that $\pratio{o}$ and $\pemp{o}$ are both increasing functions of
$\score{o}$.  Hence, the theorem states that eventually both measures are
increasing functions of $\df{\opt{o, p}}$. Such a behavior is reasonable yet
surprising since the BIC penalty vanishes from $\score{o}$ as the amount of
data increases. The reason for such behavior is that the difference between the
BIC penalty terms will outweight the difference between the likelihoods. 

\section{Experiments}
\label{sec:experiments}

In this section we will describe our empirical results\footnote{Implementation is available at \url{http://adrem.ua.ac.be/implementations}}. We begin by showing
with synthetic datasets that our measures do give the expected results. Measures
$\pemp{o}$ and $\pratio{o}$ are high for datasets in which there are no
particular order structure. On the other hand, measures are small for datasets in
which there is a clear order structure. We continue by studying the asymptotic
behavior of the score as described in Section~\ref{sec:asymptotic}.

We also tested the spectral methods with real-world data\-sets and show that all these
datasets do have an order structure. Finally, we study how well the greedy
method improves the scores of the spectral orders.

Since Atkins' algorithm for discovering spectral orders returns PQ-tree which
may correspond to several spectral orders, we sampled $1000$ random orders from
PQ-tree. However, if the number of orders represented by PQ-tree was less than
1000 we computed all the orders. 

\subsection{Synthetic datasets}
For testing purposes we considered $4$ different generated datasets. Each
dataset had $20$ attributes and $2000$ of rows. Each dataset was split into two
subsets, each having $1000$ rows. The first part of the data was used for
finding the spectral orders and the second part for actually computing the
score.  The measures $\pemp{o}$ and $\pratio{o}$ were computed by comparing them
with $1000$ random orders.

Our first dataset, \emph{Ind}, contains independent attributes. The second data,
\emph{Clust}, contains two clusters, each of $500$ rows. The attributes within
the cluster are independent. The probability of 1 was set to $3/4$ in the
first cluster and $1/4$ in the second cluster. Our third dataset, \emph{Path},
was generated such that attribute $a_i$ was generated from $a_{i-1}$ by adding
$1/4$ amount of noise, that is, $\prob{a_i = 0; a_{i - 1} = 1} = \prob{a_i = 1;
a_{i - 1} = 0} = 1/4$. The first attribute $a_1$ was generated by a fair coin
flip. Our last dataset is similar, \emph{Npath}, to \emph{Path} except that
$3/4$ were used as the amount of noise. Note that the attributes in \emph{Path}
are positively correlated whereas the consecutive attributes in \emph{Npath}
are negatively correlated.  Our expectation is that \emph{Ind} and \emph{Clust}
have no extraordinary order whereas in \emph{Path} and \emph{Npath} the
generating order is the most natural one.

\begin{table*}[htb!]
\centering
\begin{tabular}{lrrrr l@{}rrrr}
\toprule
& \multicolumn{4}{c}{$\pemp{o}$} & & \multicolumn{4}{c}{$\pratio{o}$} \\
\cmidrule{2-5}
\cmidrule{7-10}
\emph{Data} & \textsc{CO} & \textsc{CS} & \textsc{MI} & \textsc{M2} && \textsc{CO} & \textsc{CS} & \textsc{MI} & \textsc{M2}  \\
\midrule
\emph{Ind} & $0.6$ & $0.6$ & $0.6$ & $0.5$   &  & $0.6$ & $0.6$ & $0.6$ & $1.0$ \\
\emph{Clust} & $0.9$ & $0.9$ & $0.7$ & $0.7$ &  & $0.1$ & $0.2$ & $0.6$ & $0.6$ \\
\emph{Path} & $0$ & $0$ & $0$ & $0$              &  & $36.1$ & $41.8$ & $41.8$ & $41.8$ \\
\emph{Npath} & $0.98$ & $0.9$ & $0$ & $0$       &  & $0.03$ & $0.2$ & $44.4$ & $44.4$ \\

\bottomrule
\end{tabular}
\caption{Measures $\pemp{o}$ and $\pratio{o}$ of the spectral orders obtained from
the synthetic datasets. The orders are explained in Section~\ref{sec:fiedler}.}
\label{tab:syntrank}
\end{table*}

From the results given in Table~\ref{tab:syntrank} we see that we get the
expected results. The datasets \emph{Ind} and \emph{Clust} both have high measure values
since these datasets have no extraordinary order. In dataset \emph{Path} all
the orders have $\pemp{o} = 0$ suggesting that there is a strong order structure.
The orders given by the spectral algorithms for \emph{Path} are all
close to the original order. For \emph{Npath} we see
that the methods \textsc{CO} and \textsc{CS} fail to find a significant order.
The reason for this is that \emph{Npath} contains negative correlations. On
the other hand, both \textsc{MI} and \textsc{M2} find a significant order and
produce significantly small measure values.

\subsection{Asymptotic Behavior}

Our next focus is to study the asymptotic behavior of the score $\score{o}$.
Theorem~\ref{thr:asymptotic} implies that asymptotically $\pemp{o}$ is an
increasing function of $\df{\opt{o, p}}$, the number of free parameters of the
model containing the generative distribution $p$. To illustrate behavior we
generated $4$ datasets using the same method we used to generate dataset
\emph{Path}.  The datasets contained $10$ attributes and varying number of
transactions.  We computed $1000$ random orders for which we computed $\pemp{o}$
Since we know the generative distribution we were able to compute
$\df{\opt{o, p}}$ directly.

\begin{figure*}[htb!]
\centering
\subfigure[$\abs{D} = 10^3$, $\rho = 0.29$]{\includegraphics[width=4.3cm]{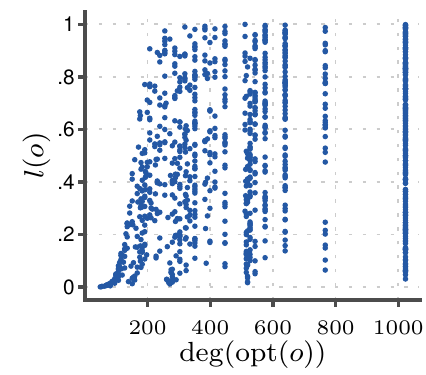}}%
\subfigure[$\abs{D} = 10^4$, $\rho = 0.77$]{\includegraphics[width=4.3cm]{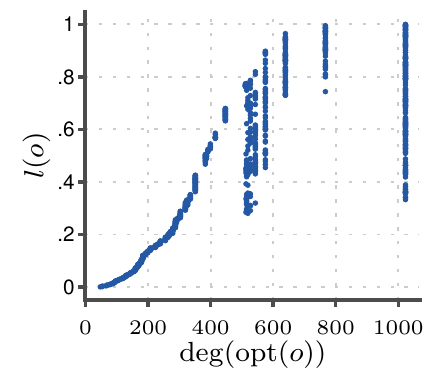}}%
\subfigure[$\abs{D} = 10^5$, $\rho = 0.97$]{\includegraphics[width=4.3cm]{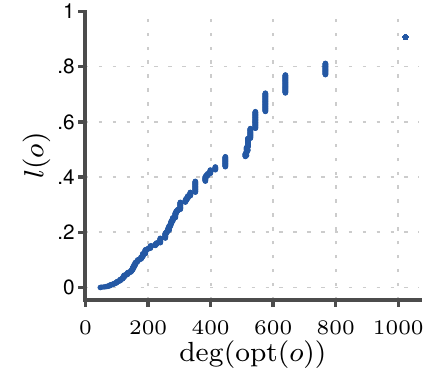}}%
\subfigure[$\abs{D} = 10^6$, $\rho = 0.98$]{\includegraphics[width=4.3cm]{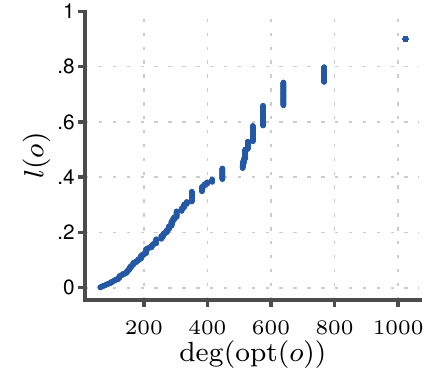}}
\caption{Measure $\pemp{o}$ as a function of $\df{\opt{o}}$, the number of free
parameters. Theorem~\ref{thr:asymptotic} implies that asymptotically $\pemp{o}$ is a
monotonic function of $\df{\opt{o}}$.  Each plot contain $1000$ random orders
of a dataset generated similarly as \emph{Path} dataset.  The number of
transactions is indicated in each sublabel. The variable $\rho$ is the
correlation between $\pemp{o}$ and $\df{\opt{o}}$.}
\label{fig:syntasympt}
\end{figure*}

From the results given in Figure~\ref{fig:syntasympt} we see that the measure
$\pemp{o}$ converges into an increasing function of $\df{\opt{o, p}}$ as the
number of transactions increases. Note that for the first two datasets there
are many orders which have the maximal number of free parameters, $1023$, yet
their measure values are small. Such behavior is hinted by
Proposition~\ref{prop:bound}: In such orders the model $\opt{o, p}$ is equal to
one segment, containing all items.  Proposition~\ref{prop:bound} states that
the necessary condition to produce this model as the model with the lowest BIC
score we must have an exponential number of transactions. Note that in the
larger datasets we have enough transactions to convince us that the model with
the worst BIC penalty term is actually the best.

\subsection{Spectral Methods with Real Datasets}

We continue our experiments with real-world datasets. The dataset
\emph{Paleo}\footnote{NOW public release 030717 available
from~\cite{fortelius05now}.}  contains information of species fossils found in
specific paleontological sites in Europe~\cite{fortelius05now}.  The dataset
\emph{Courses} contains the enrollment records of students taking courses at
the Department of Computer Science of the University of Helsinki. We took
datasets \emph{Anneal} and \emph{Mushroom} from the LUCS/KDD
repository~\cite{coenen03library}. A click-stream dataset
\emph{WebView-1}~\footnote{\url{http://www.ecn.purdue.edu/KDDCUP/data/BMS-WebView-1.dat.gz}}
was contributed by Blue Martini Software as the KDD Cup 2000
data~\cite{kohavi00bms}. The final dataset, \emph{Dna}, is DNA copy number
amplification data collection of human neoplasms~\cite{myllykangas06dna}.  Each
dataset was split into two, the first part was used for calculating the
order and the second part to calculate the actual score. To compute
the measures we also computed the scores for $1000$ random orders. The basic
characteristics and the running times are given in Table~\ref{tab:realbasic}.

\begin{table}[htb!]
\centering
\begin{tabular}{lrrrr}
\toprule
Name & $K$ & $\abs{D}$ & \% of 1s & Time  \\
\midrule
\emph{Anneal} & $73$ & $898$ & $20\%$ & $3ms$ \\
\emph{Courses} & $98$ & $3506$ & $5\%$ & $8ms$ \\
\emph{Dna} & $391$ & $4587$ & $1\%$ & $24ms$ \\
\emph{Mushroom} & $90$ & $8124$ & $25\%$ & $44ms$ \\
\emph{Paleo} & $139$ & $501$ & $5\%$ & $3ms$ \\
\emph{WebView-1} & $497$ & $59602$ & $1\%$ & $331ms$ \\
\bottomrule
\end{tabular}
\caption{Statistics and running times of datasets used in experiments. The 4th
column is the time needed to compute a score for one order.}
\label{tab:realbasic}
\end{table}

Measure $\pemp{o}$ was $0$ for all orders and datasets, except for
\emph{Anneal}, where $\pemp{\textsc{CO}} = \pemp{\textsc{CS}} = 0.03$. This
suggests that the almost all found orders were significantly good. To
illustrate this further we plotted the scores of the random and the spectral
orders in a box plot in Figure~\ref{fig:realbox}.

\begin{table}[htb!]
\centering
\begin{tabular}{lrrrr}
\toprule
\emph{Data} & \textsc{CO} & \textsc{CS} & \textsc{MI} & \textsc{M2}  \\
\midrule
\emph{Anneal} & $5.93$ & $6.04$ & $82.96$ & $47.89$ \\
\emph{Courses} & $55.47$ & $72.80$ & $54.31$ & $58.72$ \\
\emph{Dna} & $\infty$ & $\infty$ & $\infty$ & $\infty$ \\
\emph{Mushroom} & $39.13$ & $18.69$ & $46.74$ & $54.18$ \\
\emph{Paleo} & $129.39$ & $109.39$ & $13.18$ & $131.01$ \\
\emph{WebView-1} & $289.08$ & $229.18$ & $561.93$ & $764.20$ \\
\bottomrule
\end{tabular}
\caption{Measure $\pratio{o}$ of the spectral orders obtained from
the real datasets. The orders are explained in Section~\ref{sec:fiedler}.}
\label{tab:realrank}
\end{table}

\begin{figure}[htb!]
\centering
\includegraphics[width=7cm]{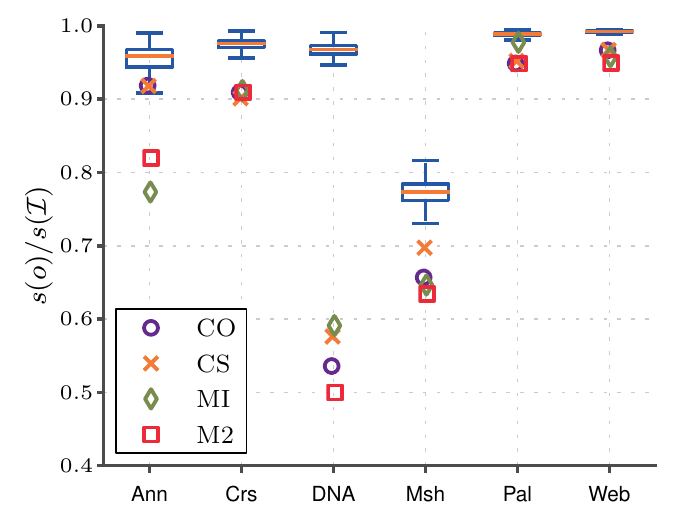}
\caption{A box plot of the scores of random and spectral orders. The scores
were normalized by dividing with $\score{\ifam{I}}$, the score of the
independent model.}
\label{fig:realbox}
\end{figure}

Examining Table~\ref{tab:realrank} reveals that the
best scores are achieved with \emph{Dna} dataset, suggesting that there is a
strong order structure in the dataset. Measures $\pratio{o}$ 
becomes infinite due to the finite precision of the floating point
number.  Also, the scores of the spectral orders for the \emph{Mushroom}
dataset are small but so are the scores of the random orders. This implies that
there are lot of dependencies in the dataset but the order structure is not
that strong. Interestingly enough, either \textsc{MI} or \textsc{M2} produces
the lowest score for $5$ datasets. One possible explanation is that \textsc{MI}
and \textsc{M2} are able to use the negative correlations and their score is
directly related to the log-likelihood of the model. On the other hand, the
order \textsc{CS} is the best for the \emph{Courses} dataset and the order
\textsc{MI} fails with \emph{Paleo} dataset.

\subsection{Improving the Scores with Greedy Search}

We conclude our experiments by studying the greedy met\-hod discussed in
Section~\ref{sec:greedy}. We applied the algorithm for the first part of each
dataset. As starting points we used the orders obtained by the spectral
methods. Our hope is that the greedy method improves the scores of spectral
order because it is able to use statistics of higher-order and because spectral
orders are only guaranteed to work with $L$-matrices (see
Section~\ref{sec:fiedler}).

For comparison we sampled up to $50$ orders from the PQ-tree produced by
Atkins' algorithm. Each order was used as a starting point for the greedy
algorithm.  We also tested the greedy method with $50$ random starting orders.
The obtained orders were then evaluated by computing the scores from the second
part of the dataset. The running times varied from $1$ second to $1$ hour,
depending on the size of the dataset.

\begin{table}[htb!]
\centering
\begin{tabular}{lrrrrr}
\toprule
Name & \textsc{CO} & \textsc{CS} & \textsc{MI} & \textsc{M2} & \textsc{RND} \\
\midrule
\emph{Anneal} & $5.05$ & $3.39$ & $2.98$ & $3.25$ & $3.41$ \\
\emph{Courses} & $2.50$ & $2.02$ & $1.56$ & $2.82$ & $1.68$ \\
\emph{Dna} & $7.36$ & $9.26$ & $12.57$ & $8.80$ & $2.31$ \\
\emph{Mushroom} & $7.28$ & $9.18$ & $9.91$ & $13.33$ & $7.66$ \\
\emph{Paleo} & $0.80$ & $0.98$ & $0.61$ & $0.98$ & $0.58$ \\
\emph{WebView-1} & $0.92$ & $1.62$ & $1.20$ & $1.59$ & $0.63$ \\
\bottomrule
\end{tabular}
\caption{Gains of the scores when using the greedy method compared to the
scores of the starting points. The percentages are computed as $100\% - 100\%
\times \score{o_1} / \score{o_2}$, where $o_1$ is the final order and $o_2$ is
the starting order.}
\label{tab:realcomp}
\end{table}

\begin{table*}[htb!]
\centering

\begin{tabular}{lr@{}r rrrrr}
\toprule
& \multicolumn{1}{c}{$\pemp{o}$} && \multicolumn{5}{c}{$\pratio{o}$} \\
\cmidrule{2-2}
\cmidrule{4-8}
\emph{Data} & \textsc{RND} && \textsc{CO} & \textsc{CS} & \textsc{MI} & \textsc{M2} & \textsc{RND}  \\
\midrule
\emph{Anneal}   & $0.12$ &  & $20.1$ & $14.6$ & $104.2$ & $66.9$ & $2.9$ \\
\emph{Courses}  & $0.10$ &  & $100.3$ & $110.9$ & $79.8$ & $109.8$ & $3.4$ \\
\emph{Dna}      & $0.03$ &  & $\infty$ & $\infty$ & $\infty$ & $\infty$ & $4.9$ \\
\emph{Mushroom} & $0.01$ &  & $75.2$ & $56.1$ & $100.9$ & $136.1$ & $7.5$ \\
\emph{Paleo}    & $0.13$ &  & $181.5$ & $169.7$ & $26.9$ & $196.3$ & $2.8$ \\
\emph{WebView-1}   & $0$ &  & $453.7$ & $570.5$ & $871.4$ & $940.0$ & $11.6$ \\
\bottomrule
\end{tabular}

\caption{Measures $\pemp{o}$ and $\pratio{o}$ of the orders obtained using the
greedy method with the spectral and random orders as the starting points.
Measure values $\pemp{o}$ for all spectral orders were $0$ and are
omitted from the table.}
\label{tab:realhill}
\end{table*}

By comparing the measure $\pratio{o}$ given in
Table~\ref{tab:realhill} to the values given in
Tables~\ref{tab:realrank} we see that the greedy
method does not perform well alone: when random orders are used as starting
points, the discovered orders are worse than the spectral orders.  However,
greedy method is useful when spectral orders are used as starting points.  From
Table~\ref{tab:realcomp} we see that the greedy method improves the scores of
the spectral orders up to $13$ percents. The gain of the score depends on
the dataset but less on the spectral method used. The scores for
\emph{Courses}, \emph{Paleo} and \emph{WebView-1} datasets improve up to $3$
percents where as the biggest gains are with \emph{Dna}, and \emph{Mushroom}
datasets where the scores improve by $7$ -- $13$ percents.

\section{Related Work}
\label{sec:related}

A popular choice for measuring the goodness of an order is the Lazarus count,
the number of $0$s between $1$s in a row. If the Lazarus count is $0$, then the
data is said to have the consecutive ones property.  In some cases this has a
natural interpretation, for example, in a paleontological data a taxon becomes
extant and then extinct. If the matrix has a consecutive
ones property, then the Fiedler vector of the co-occurrence matrix returns the
correct order~\cite{atkins99seriation}. It is an open question why the spectral
method works also with the noisy data. An alternative approach has been
suggested in~\cite{puolamaki06seriation}, where the authors construct a
probabilistic model encapsulating the consecutive ones property.

Ranking or sorting items can be seen as deducing a linear order for the items.
Applications for ranking are, for example, finding relevant web
pages~\cite{page98page,kleinberg99hits} or ranking database query
results~\cite{fagin04comparing}. One of the key differences in these approaches
and ours is that in our case the reversed order is as good as the original.  We
are interested in finding the order in which the dependent attributes are
close.  This goal is different than finding the most relevant items.

In some cases, linear orders is too strict a structure, in such case partial
orders (transitive, asymmetric, and reflexive relation) may be more natural.
For example, consider the course enrollment data, in which the same basic
course is prerequisite for several advanced independent courses.  Finding
partial orders have been studied for example in~\cite{ukkonen05partial}.
General partial orders seem to be very complex objects. A simple but yet
interesting subclass of partial orders are bucket
orders~\cite{gionis06bucket,fagin04comparing}. A problem of searching fragments
of order, that is finding a collection of linear orders defined for a subset of
items has been studied in~\cite{gionis03fragments}.

\section{Conclusions}
\label{sec:conclusion}
We studied the concept of measuring the goodness of an order. We say that the
order is good if the heavily dependent attributes are close to each other.  In
order to define the score we introduce an order-sensitive model and then use
the BIC score to rank the model.
To find the optimal model we created a dynamic program and show that it can be
evaluated in $O\fpr{K\min\fpr{K, \log \abs{D}}^2}$ time. Hence, our method
works well even for the datasets with vast number of items.

We provided asymptotic results showing that the score is connected the number
of free parameters in the model. We also demonstrate this result empirically
with synthetic data.

We compared the score of the order against the scores of random orders. We say
that the order is good if the score is exceptionally lower than the score of
the random order. We used two different measures, $\pemp{o}$ the proportion of
random scores having the smaller score than $o$, and $\pratio{o}$, the ratio of
the score $\score{o}$ and the average score of a random order. One of our
future goals is to develop a more refined measure for comparing the score
against the scores of random orders.

We evaluate the measures with several spectral and greedy methods. In our
experiments we found out that Fiedler orders of the mutual information matrices
(see Section~\ref{sec:fiedler}) produced better results for our datasets than
the orders based on co-occurrences or cosine distance.  In our experiments, the
greedy optimization improved the scores of spectral orders up to $13$ percents.

One of our future goals is to extend the current method for more general
partial orders (see Section~\ref{sec:related}). Such an extension is not
trivial since in general case finding the best model is a computationally
difficult task.

\bibliographystyle{IEEEtran}
\bibliography{order}  

\begin{thebibliography}{10}
\providecommand{\url}[1]{#1}
\csname url@samestyle\endcsname
\providecommand{\newblock}{\relax}
\providecommand{\bibinfo}[2]{#2}
\providecommand{\BIBentrySTDinterwordspacing}{\spaceskip=0pt\relax}
\providecommand{\BIBentryALTinterwordstretchfactor}{4}
\providecommand{\BIBentryALTinterwordspacing}{\spaceskip=\fontdimen2\font plus
\BIBentryALTinterwordstretchfactor\fontdimen3\font minus
  \fontdimen4\font\relax}
\providecommand{\BIBforeignlanguage}[2]{{%
\expandafter\ifx\csname l@#1\endcsname\relax
\typeout{** WARNING: IEEEtran.bst: No hyphenation pattern has been}%
\typeout{** loaded for the language `#1'. Using the pattern for}%
\typeout{** the default language instead.}%
\else
\language=\csname l@#1\endcsname
\fi
#2}}
\providecommand{\BIBdecl}{\relax}
\BIBdecl

\bibitem{fortelius06spectral}
M.~Fortelius, A.~Gionis, J.~Jernvall, and H.~Mannila, ``Spectral ordering and
  biochronology of european fossil mammals,'' \emph{Paleobiology}, vol.~32,
  no.~2, pp. 206--214, March 2006.

\bibitem{kendall04abudance}
D.~G. Kendall, ``Abundance matrices and seriation in archaeology,''
  \emph{Probability Theory and Related Fields}, vol.~17, no.~2, pp. 104--112,
  June 1971.

\bibitem{gionis04geometric}
A.~Gionis, H.~Mannila, and J.~K. Sepp\"{a}nen, ``Geometric and combinatorial
  tiles in 0-1 data,'' in \emph{8th European Conference on Principles and
  Practice of Knowledge Discovery in Databases (PKDD 2004)}, 2004, pp.
  173--184.

\bibitem{cowell99network}
R.~G. Cowell, A.~P. Dawid, S.~L. Lauritzen, and D.~J. Spiegelhalter,
  \emph{Probabilistic Networks and Expert Systems}, ser. Statistics for
  Engineering and Information Science, M.~Jordan, S.~L.~L. nad Jeral
  F.~Lawless, and V.~Nair, Eds.\hskip 1em plus 0.5em minus 0.4em\relax
  Springer-Verlag, 1999.

\bibitem{schwarz78bic}
G.~Schwarz, ``Estimating the dimension of a model,'' \emph{Annals of
  Statistics}, vol.~6, no.~2, pp. 461--464, 1978.

\bibitem{papadimitriou94complexity}
C.~H. Papadimitriou, \emph{Computational Complexity}.\hskip 1em plus 0.5em
  minus 0.4em\relax Addison-Wesley, 1994.

\bibitem{fiedler75fiedler}
M.~Fiedler, ``A property of eigenvectors of nonnegative symmetric matrices and
  its application to graph theory,'' \emph{Czechoslovak Mathematical Journal},
  vol.~25, no.~4, pp. 619--633, 1975.

\bibitem{atkins99seriation}
J.~E. Atkins, E.~G. Boman, and B.~Hendrickson, ``A spectral algorithm for
  seriation and the consecutive ones problem,'' \emph{SIAM J. Comput.},
  vol.~28, no.~1, pp. 297--310, 1999.

\bibitem{fortelius05now}
M.~Fortelius, ``Neogene of the old world database of fossil mammals ({NOW}),''
  University of Helsinki, \url{http://www.helsinki.fi/science/now/}, 2005.

\bibitem{coenen03library}
F.~Coenen, ``The {LUCS}-{KDD} discretised/normalised {ARM} and {CARM} data
  library,'' 2003.

\bibitem{kohavi00bms}
R.~Kohavi, C.~Brodley, B.~Frasca, L.~Mason, and Z.~Zheng, ``{KDD-Cup} 2000
  organizers' report: Peeling the onion,'' \emph{SIGKDD Explorations}, vol.~2,
  no.~2, pp. 86--98, 2000.

\bibitem{myllykangas06dna}
S.~Myllykangas, J.~Himberg, T.~Böhling, B.~Nagy, J.~Hollm\'{e}n, and
  S.~Knuutila, ``Dna copy number amplification profiling of human neoplasms,''
  \emph{Oncogene}, vol.~25, no.~55, pp. 7324--7332, Nov. 2006.

\bibitem{puolamaki06seriation}
\BIBentryALTinterwordspacing
K.~Puolam\"aki, M.~Fortelius, and H.~Mannila, ``Seriation in paleontological
  data using markov chain monte carlo methods,'' \emph{PLoS Comput Biol},
  vol.~2, no.~2, Feb 2006. [Online]. Available:
  \url{http://dx.doi.org/10.1371%2Fjournal.pcbi.0020006}
\BIBentrySTDinterwordspacing

\bibitem{page98page}
S.~Brin and L.~Page, ``The anatomy of a large-scale hypertextual web search
  engine,'' \emph{Computer Networks and ISDN Systems}, vol.~30, no. 1--7, pp.
  107--117, 1998.

\bibitem{kleinberg99hits}
J.~Kleinberg, ``Authoritative sources in a hyperlinked environment,''
  \emph{Journal of the ACM}, vol.~46, no.~5, pp. 604--632, 1999.

\bibitem{fagin04comparing}
R.~Fagin, R.~Kumar, M.~Mahdian, D.~Sivakumar, and E.~Vee, ``Comparing and
  aggregating rankings with ties,'' in \emph{Proceedings of the ACM-SIAM
  Symposium on Discrete Algorithms (PODS)}, 2004, pp. 47--58.

\bibitem{ukkonen05partial}
A.~Ukkonen, M.~Fortelius, and H.~Mannila, ``Finding partial orders from
  unordered 0-1 data,'' in \emph{KDD '05: Proceedings of the eleventh ACM
  SIGKDD international conference on Knowledge discovery in data mining}.\hskip
  1em plus 0.5em minus 0.4em\relax New York, NY, USA: ACM, 2005, pp. 285--293.

\bibitem{gionis06bucket}
A.~Gionis, H.~Mannila, K.~Puolam\"{a}ki, and A.~Ukkonen, ``Algorithms for
  discovering bucket orders from data,'' in \emph{KDD '06: Proceedings of the
  12th ACM SIGKDD international conference on Knowledge discovery and data
  mining}.\hskip 1em plus 0.5em minus 0.4em\relax New York, NY, USA: ACM, 2006,
  pp. 561--566.

\bibitem{gionis03fragments}
A.~Gionis, T.~Kujala, and H.~Mannila, ``Fragments of order,'' in \emph{KDD '03:
  Proceedings of the ninth ACM SIGKDD international conference on Knowledge
  discovery and data mining}.\hskip 1em plus 0.5em minus 0.4em\relax New York,
  NY, USA: ACM, 2003, pp. 129--136.

\bibitem{csiszar75divergence}
I.~Csisz\'ar, ``I-divergence geometry of probability distributions and
  minimization problems,'' \emph{The Annals of Probability}, vol.~3, no.~1, pp.
  146--158, Feb. 1975.

\bibitem{wald43test}
A.~Wald, ``Tests of statistical hypotheses concerning several parameters when
  the number of observations is large,'' \emph{Trans. of the American
  Mathematical Society}, vol.~54, no.~3, pp. 426--482, Nov. 1943.

\bibitem{vaart98statistics}
A.~W. van~der Vaart, \emph{Asymptotic Statistics}, ser. Cambridge Series in
  Statistical and Probabilistic Mathematics.\hskip 1em plus 0.5em minus
  0.4em\relax Cambridge University Press, 1998.

\end{thebibliography}

\appendix
\section{Proof of Theorem~\ref{thr:asymptotic}}

Throughout the whole section we will assume that the underlying distribution
$p$ has no zero probabilities, $p > 0$.

Given a distribution $q$ and a set of segments $\sfam{C}$ we will adopt the
notation $h(A = t) = q(A = t; \sfam{C})$ to mean the unique distribution $h \in
\model{\sfam{C}}$ such that $q(C) = h(C)$ for any $C \in \sfam{C}$.

We begin by considering indicator functions. Given an itemset $X =
\enset{a_{i_1}}{a_{i_L}}$ we define an indicator function
$\funcdef{S_X}{\set{0,1}^K}{\set{0,1}}$ to be $S_X(t) = t_{i_1}\cdots t_{i_L}$.
Thus $S_X(t) = 1$ if and only if all $t_i$ corresponding to the elements of $X$
have value of 1. We also define $S_\emptyset(t) = 1$.

\begin{lemma}
\label{lem:basis}
Given a family of itemsets $\ifam{X}$, the set of indicator functions
$\mathcal{S} = \set{S_X ; X \in \ifam{X}}$ are independent, that is, $\sum_{X
\in \ifam{X}} r_X S_X(t) = 0$ is possible only for $r_X = 0$. If $\ifam{X}$
contains all itemsets, then $\mathcal{S}$ forms a basis.
\end{lemma}

\begin{proof}
It is sufficient to prove the lemma for the case when $\ifam{X} = \set{X; X \in
A}$ contains all possible itemsets.

For a $v \in \set{0, 1}^K$, Consider a function
$\funcdef{T_v}{\set{0,1}^K}{\set{0,1}}$ for which $T_v(t) = 1$ if and only if
$t = v$. Obviously, the set of functions $\mathcal{T} = \set{T_v; v \in \set{0,
1}^K}$ is linearly independent. Hence to prove the lemma we need to show that
we can build a linear mapping from $\mathcal{T}$ onto $\mathcal{S}$ that has an
inverse.

To prove this we first note that $S_X(t) = \sum_{Y \supseteq X} T_Y(t)$. This
can be inverted using the inclusion-exclusion principle $T_X(t) = \sum_{Y
\supseteq X} (-1)^{\abs{Y} - \abs{X}}S_Y(t)$. This completes the proof.
\end{proof}

A direct calculation implies that the distributions in $\model{\sfam{C}}$ have
a particular exponential form.

\begin{lemma}
\label{lem:exponential}
Let $\sfam{X} = \set{X; X \subseteq C \in \sfam{C}}$ be the downward closure of
$\sfam{C}$. A distribution $p$ is in $\model{\sfam{C}}$ if only if $p(A =
t) = \exp\fpr{\sum_{X \in \sfam{X}} r_X S_X(t)}$ for some specific constants
$\set{r_X}$.
\end{lemma}

\begin{proof}
Assume that $p \in \model{\sfam{C}}$. Then $p$ has a form of $p(A) = \prod
p\fpr{C_i} / \prod p\fpr{C_i \cap C_{i - 1}}$. Fix $i$.  Lemma~\ref{lem:basis}
now implies that we can write $\log p\fpr{C_i = t_{C_i}} = \sum_{X \subseteq
C_i} a_X S_X(t)$ or in other words $p\fpr{C_i = t_{C_i}} = \exp\fpr{\sum_{X
\subseteq C_i} a_X S_X(t)}$. We can perform this factorization for each
$p(C_i)$ and for each $p\fpr{C_i \cap C_{i - 1}}$. By combining these
factorizations we arrive at the desired result.

We will prove the other direction by induction. Let $p$ be the distribution
having the exponential form and let $\set{r_X}$ be the needed constants. Shorten
$V = C_1$, $W = A - (C_1 - C_2)$, and $H = C_1 \cap C_2$.

Using the
argument in the first paragraph we can write $p\fpr{H = t_H} = \exp\fpr{\sum_{X
\subseteq H} u_X S_X(t)}$ for some constants $\set{u_X}$.
Note that if $X \in \sfam{X}$ is such that $X \nsubseteq H$, then either $X \subseteq V$
or $X \subseteq W$.
We have
\[
\begin{split}
&\log p\fpr{A - H = t_{A - H} \mid H = t_H} \\
& \quad =\sum_{X \in \sfam{X}} r_X S_X(t) - \sum_{X \in \sfam{X}, X \subseteq H} u_XS_X(t) \\
& \quad =\sum_{X \in \sfam{X}, X \subseteq V} r_X S_X(t) + \sum_{X \in \sfam{X}, X \subseteq W} r_X S_X(t) \\
& \qquad - \sum_{X \in \sfam{X}, X \subseteq H} (r_X + u_X)S_X(t).
\end{split}
\]

The third term depends only on $t_H$, hence it is a part of the
normalization constant of the conditioned distribution. The first term depends only
on the values of $t_V$ while the second term depends only on the values of $t_W$.
This leads to
\[
\begin{split}
p(A) & = p(A - H \mid H)p(H) \\
& = p(V - H \mid H)p(W - H \mid H)p(H)\\
& = \frac{p(V)p(W)}{p(H)} = \frac{p(C_1)p(W)}{p(C_1 \cap C_2)}.
\end{split}
\]
The distribution $p(W) = p(W \mid H)p(H)$ has an exponential form, hence
if we assume inductively that $p(W) \in \model{\sfam{C} - C_1}$, then the
above equality proves that $p(A) \in \model{\sfam{C}}$.
\end{proof}

Our next step is to study $\opt{o, p}$. Let $\sfam{V}, \sfam{W} \in \seg{o}$.
We say that $\sfam{W}$ \emph{refines} $\sfam{V}$ if for each $V \in \sfam{V}$
there is $W \in \sfam{W}$ such that $V \subseteq W$. In this case we write
$\sfam{V} \sqsubset \sfam{W}$. We will now show that $\opt{o, p}$ is the
minimal set with the respect of the relation $\sqsubset$.

\begin{theorem}
\label{thr:optseg}
Let $\sfam{O} = \opt{o, p}$. If $\sfam{C}$ is such that $p \in
\model{\sfam{C}}$, then $\sfam{O} \sqsubset \sfam{C}$. Moreover,
$\sfam{O}$ is unique.
\end{theorem}

\begin{proof}
Assume otherwise and define $\sfam{V}$ to be the maximal segments of the set
of segments
$\set{C \cap O; C \in \sfam{C}, O \in \sfam{O}}$. Clearly $\df{\sfam{V}} <
\df{\sfam{O}}$. We need to show that $p \in \model{\sfam{V}}$. Let $\sfam{X}$
and $\sfam{Y}$ be the downward closure of $\sfam{O}$ and $\sfam{C}$,
respectively. Let $\set{r_X}$, $\set{s_Y}$ be the sets of constants in
Lemma~\ref{lem:exponential} when applied to $\ifam{X}$ and $\ifam{Y}$,
respectively. We have
\[
	\sum_{X \in \sfam{X}} r_X S_X(t) = \log p(A = t) = \sum_{Y \in \sfam{Y}} s_Y S_Y(t).
\]
Since the functions $S_X$ are linearly independent we must have that $r_Z =
s_Z$ for any $Z \in \sfam{X} \cap \sfam{Y}$ and $r_X = 0$, $s_Y = 0$ otherwise.
The downward closure of $\sfam{V}$ is exactly $\sfam{X} \cap \sfam{Y}$.  But
now Lemma~\ref{lem:exponential} implies that $p \in \model{\sfam{V}}$ which is
a contradiction by the minimality of $\sfam{O}$.
\end{proof}

We will need the following technical lemma that states that the difference
between the likelihood terms is bounded.

\begin{lemma}
\label{lem:bound}
Assume two orders $o_1$ and $o_2$. Let $\sfam{O}_i = \opt{o_i, p}$ and let
$\sfam{C}_i \in \seg{o_i}$ such that $\sfam{O}_i \sqsubset \sfam{C}_i$. Let $D$
be a dataset with $N = \abs{D}$ transactions and let $h_i = q_D\pr{A;
\sfam{C}_i} \in \model{\sfam{C}_i}$. Let $Z_N = N\ent{h_1} - N\ent{h_2}$
and let $Z$ be the limit distribution as $N \to \infty$.  Then $P(Z_N < t) \to P(Z < t)$ uniformly in $t$ as $N \to \infty$. Moreover, $P(Z = \infty) = 0$.

\end{lemma}

\begin{proof}
Let $\sfam{X}_i$ be the downward
closure of $\sfam{O}_i$ and let $\set{r^i_O}$ be the set of constraints given
by Lemma~\ref{lem:exponential} when applied to $\sfam{O}_i$ and $p$. Note
that we have
\[
	\sum_{O \in \sfam{X}_1} r^1_O S_O(t) = \log p(A = t) = \sum_{O \in \sfam{X}_2} r^2_O S_O(t).
\]
Since the functions $S_O$ are linearly independent we must have $r^1_O = r^2_O$
whenever $O \in \sfam{X}_1 \cap \sfam{X}_2$ and $r^i_O = 0$ otherwise.

Let $\sfam{Y} = \sfam{X}_1 \cap \sfam{X}_2$. We know that
(see Theorem~3.1 in \cite{csiszar75divergence}) there is a set of constants
$\set{r_Y}$ such that the distribution $g$ defined as $g(A = t) =
\exp\fpr{\sum_{Y \in \sfam{Y}} r_YS_Y(t)}$ has the property $g(S_Y = 1) =
q_D(S_Y = 1)$. Now a classic result (see \cite{wald43test} for example)
implies that $N\ent{h_i} - N\ent{g}$ converges into a chi-square distribution.
A known result (see Theorem~2.3 in \cite{vaart98statistics}) implies that the
distribution of $Z$ is a difference of two chi-square distributions implying
that $P(Z = \infty) = 0$. Since $Z$ is continuous,  a known result
(see Lemma~2.11 in {vaart98statistics}) now implies that the convergence is
uniform.
\end{proof}

Elementary analysis now implies the following key corollary.
\begin{corollary}
\label{cor:limit}
Let $Z_N$ as in Lemma~\ref{lem:bound} and let $t_N$ be a sequence such that
$t_N \to \infty$. Then $\prob{Z_N \geq t_N} \to 0$.
\end{corollary}

Let $D$ be a dataset with $N$ transactions.  We will denote by $\sfam{C}^N
\in \seg{o}$ to be the model with the lowest score. Note that $\sfam{C}^N$
is a random variable since it depends on $D$. We will show that $\sfam{C}^N$
converges into $\opt{o, p}$.

\begin{theorem}
\label{thr:converge}
Let $\sfam{O} = \opt{o, p}$. Then $P\pr{\sfam{O} = \sfam{C}^N} \to 1$
as $N$ approaches infinity.
\end{theorem}

\begin{proof}
Let $\sfam{C} \in \seg{o}$ such that $\sfam{O} \not\sqsubset \sfam{C}$.  Let $h
= p(A; \sfam{C}) \in \model{\sfam{C}}$.  Theorem~\ref{thr:optseg} now implies
that $h \neq p$. Note that 
\[
\begin{split}
\lim_N \frac{\score{\sfam{C}}}{N} &= -\sum_t p(A=t)\log h(A=t) \\
&> \ent{p} = \lim_N \frac{\score{\sfam{O}}}{N}.
\end{split}
\]
Hence $P\pr{\sfam{C} = \sfam{C}^N} \to 0$ since $\sfam{C}^N$ is the
collection of segments with the optimal score.

Assume now that $\sfam{O} \sqsubset \sfam{C}$, $\sfam{O} \neq \sfam{C}$. Let $h
= q_D\fpr{A; \sfam{O}}$ and $g = q_D\fpr{A; \sfam{C}}$. Set $d =
1/2(\df{\sfam{C}} - \df{\sfam{O}})$. Note that
\[
\score{\sfam{O}} - \score{\sfam{C}} = N\ent{h} - N\ent{g} - d\log N.
\]
Since $d\log N \to \infty$ as $N \to \infty$, Corollary~\ref{cor:limit}
implies that
\[
\prob{\score{\sfam{O}} \geq \score{\sfam{C}}} = \prob{N\ent{h} - N\ent{g} \geq d\log N} \to 0.
\]
Thus we must have $P\pr{\sfam{O} = \sfam{C}^N} \to 1$.
\end{proof}

We are now ready to prove the main theorem.
\begin{proof}[of Theorem~\ref{thr:asymptotic}]
The theorem follows if we can prove that $\score{o_1}$ is lower than
$\score{o_2}$ with probability $1$ as $N$ approaches infinity. Let
$\sfam{O}_i = \opt{o_i, p}$. We write
\[
\begin{split}
&P\pr{\score{o_1} \geq \score{o_2}} \\
& \quad \leq P\pr{\score{\sfam{O}_1} \geq \score{\sfam{O}_2}} + P\pr{\sfam{C}^N_1 \neq \sfam{O}_1 \lor \sfam{C}^N_2 \neq \sfam{O}_2} \\
& \quad \leq P\pr{\score{\sfam{O}_1} \geq \score{\sfam{O}_2}} + P\pr{\sfam{C}^N_1 \neq \sfam{O}_1} + P\pr{\sfam{C}^N_2 \neq \sfam{O}_2}.
\end{split}
\]
Theorem~\ref{thr:converge} implies that the second and the third terms converge
to 0. To annihilate the first term we use again Corollary~\ref{cor:limit}. Let
$g_i = q_D\fpr{A; \sfam{O}_i}$. Set $d = 1/2(\df{\sfam{O}_2} -
\df{\sfam{O}_1})$. Note that
\[
\score{\sfam{O}_1} - \score{\sfam{O}_2} = N\ent{g_1} - N\ent{g_2} - d\log N.
\]
Since $d\log N \to \infty$ as $N \to \infty$, Corollary~\ref{cor:limit}
implies that
\[
\begin{split}
&\prob{\score{\sfam{O}_1} \geq \score{\sfam{O}_2}}  \\
&\quad = \prob{N\ent{g_1} - N\ent{g_2} \geq d\log N} \to 0.
\end{split}
\]
This completes the proof.
\end{proof}

\end{document}